\documentclass[11pt,reqno]{amsart}
\usepackage[utf8]{inputenc}
\usepackage{braket,amsmath,amsfonts,amsthm,url,graphicx,url,xcolor,hyperref,subcaption,appendix,comment,physics}

\newcommand{\id}{\hat{1}}

\newtheorem{lemma}{Lemma}[section]
\newtheorem{prop}[lemma]{Proposition}
\newtheorem{theorem}[lemma]{Theorem}

\newtheorem{cor}[lemma]{Corollary}

\newtheorem{rem}[lemma]{Remark}

\newtheorem{exam}[lemma]{Example}

\newcommand{\Hil}{\mathcal{H}}

\newcommand{\NN}{\mathbb{N}}
\newcommand{\MM}{\mathbb{M}}
\newcommand{\RR}{\mathbb{R}}

\newcommand{\E}{\mathcal{E}}
\renewcommand{\S}{\mathcal{S}}

\newcommand{\pl}{\hspace{.1cm}}

\renewcommand{\L}{\mathcal{L}}

\newcommand{\Ha}{H}
\newcommand{\epow}[2]{F^{({#2})}_{{#1}}}

\newcommand{\Id}{\text{Id}}

\DeclareUnicodeCharacter{2009}{\,}

\title{Information Fragility or Robustness Under Quantum Channels}

\author[N. LaRacuente]{Nicholas LaRacuente}
\address{Nicholas LaRacuente, Department of Computer Science, University of Indiana Bloomington, 700 N Woodlawn Ave, Bloomington, IN, 47408, USA}
\author[G. Smith]{Graeme Smith}
\address{Graeme Smith, Institute for Quantum computing and Department of Applied Mathematics, University of Waterloo, 200 University Ave W, Waterloo, ON N2L 3G1, Canada}

\thanks{\hspace{-4.4mm}GS was supported by NSF grants CCF 1652560 and PHY 1915407.  This work was supported by the Canada First Research Excellence Fund (CFREF). NL was supported by IBM as a Postdoctoral Scholar at the University of Chicago \& the Chicago Quantum Exchange. NL is supported by Indiana University Bloomington. Correspondence: nlaracu@iu.edu}

\begin{document}

\begin{abstract}
    Quantum states naturally decay under noise. Many earlier works have quantified and demonstrated lower bounds on the decay rate, showing exponential decay in a wide variety of contexts. Here we study the converse question: are there uniform upper bounds on the ratio of post-noise to initial information quantities when noise is sufficiently weak?
    In several scenarios, including classical, we find multiplicative converse bounds.  However, this is not always the case.  Even for simple noise such as qubit dephasing or depolarizing, mutual information may fall by an unbounded factor under arbitrarily weak noise. As an application, we find families of channels with non-zero private capacity despite arbitrarily high probability of transmitting an arbitrarily good copy of the input to the environment. 
\end{abstract}

\maketitle

\section{Introduction}
Common intuition suggests that in the absence of error correction, quantum states exposed to noise decay exponentially. A great deal of effort has gone in to proving and formalising this intuition, culmintating in universal decay inequalities that are uniform in the input state. In particular, recall the quantum relative entropy denoted $D(\rho \| \sigma) = \Tr \rho(\log \rho -\log\sigma)$ for a pair of density matrices $\rho$ and $\sigma$. It was shown in \cite{gao_complete_2022} or by combining \cite{gao_geometric_2021} with \cite{junge_stability_2022} that within finite dimension and under the right detailed balance condition, continuously generated families of quantum channels admit universal exponential decay. A substantial update \cite{gao_complete_2022-1} further related the decay rate constants more precisely to notions of spectral gap and subspace inclusion indices, obtaining conjectured optimal decay rates in terms of these quantities. Certain forms of exponential decay also extend to semigroups that lack detailed balance \cite{laracuente_self-restricting_2022}. As applications, relative entropy decay rates upper bound quantities such as quantum, private, and classical channel capacity \cite{bardet_group_2021}. In particular, all of these bounds upper bound the \textit{ratio} between final and initial information quantities, predicting non-trivial decay no matter how small the initial relative entropy.

In this work, we study the converse: is there any lower bound on the ratio between an information quantity before and after a small amount of noise? One may for instance ask if an arbitrarily small amount of depolarizing or dephasing noise can eliminate an arbitrarily large portion of the intial mutual information between a system and an untouched auxiliary. For quantum channels, the answer is yes: even for simple channels, there need not be a uniform lower bound on the fraction of the initial information that remains after applying a noise channel, no matter how small the noise.  In other words, there are bipartite states whose correlations are so delicate that even a tiny amount of noise can destroy nearly all of the initial mutual information present. 

Here we use the symbol $I[A:B](\rho)$ to denote the mutual information between subsystems $A$ and $B$ for an input density $\rho$, avoiding the more common subscript notation due to proliferation of sub and superscripts. We recall the Lindbladian and quantum Markov semigroup formalism, in which a family of quantum channels $(\Phi^t)_{t \in \RR^+}$ is generated by a Lindbladian $\L$ via the relation $\Phi^t = \exp(- t \L)$. The reason to study Lindbladians is the ability to take the noise to arbitrarily small levels by taking $t$ close to zero. Many common classes of channels, including depolarizing and dephasing, can be written using re-parameterized Lindbladians. Though we make use of Lindbladians in describing our results, the settings we consider are closer to information theory than to traditional studies of open quantum systems. We provide several circumstances in which the ratio between initial information and decayed information is lower bounded, at times uniformly in the input density.  A summary of these situations can be found in the following theorem.

\begin{theorem}[Decay Converse] \label{thm:converse}
Let $\Phi^t = \exp( - t \L)$ be any quantum Markov semigroup acting on $B$ within bipartite system $A \otimes B$, generated by Lindbladian $\L$ with fixed point projection $\E$.
\begin{enumerate}
    \item Assume $\L$ has fixed point projection $\E$. Then there exists a function $f_{1:(\L)}$ depending on $\L$ for which
    \[ D(\Phi^t(\rho) \| \E(\rho)) \geq \big ( 1 - f_{1:(\L)}(t) \big ) D(\rho \| \E(\rho)) \]
    for every input density $\rho$.
    \item Assume that $\E(\rho) = \E(\sigma)$ and that $\rho, \sigma$, and $\E(\rho)$ commute. Then there exists a function $f_{2:(\L, \sigma)}$ depending on $\L$ and $\sigma$ for which
    \[ D(\Phi^t(\rho) \| \Phi^t(\sigma)) \geq \big ( 1 - f_{2:(\L, \sigma)}(t) \big ) D(\rho \| \sigma ) \pl. \]
    \item Assume that $A$ and $B$ are classical probability spaces. Then there exists a function $f_{3:(\L, \rho^A \otimes \rho^B)}$ depending on $\L$ and $\rho^A \otimes \rho^B$ for which
    \[ I[A:B]((\id^A \otimes \Phi^\lambda)(\rho)) \geq \big (1 -f_{3:(\L, \rho^A \otimes \rho^B)}(t) \big ) I[A:B](\rho) \]
    for every density $\rho$ on $A \otimes B$.
\end{enumerate}
In each case as above, $f_{...} : \RR^+ \rightarrow [0,1)$, and $f_{...}(0) = 0$.
\end{theorem}
Part (1) follows Theorem \ref{thm:clsiconverse}. Parts (2) and (3) follow Theorem \ref{thm:classical} and Corollary \ref{cor:classicalmut}. In the following Section \ref{sec:sudden}, we show how bounds of this form are violated by quantum channels and input states. 

After an illustrative example on qubit channels, we show that it is fairly common for mutual information to decay very quickly, at least on certain input states.  This is summarized by the following theorem.

\begin{theorem} \label{thm:group}
Let $\L$ be a finite group channel generator given by
\[ \L(\rho) = \rho - \sum_{j = 1}^k \big ( \frac{1}{2} p_j u_j \rho u_j^\dagger + \frac{1}{2}  p_j u_j^\dagger \rho u_j \big )  \]
for some representation $\{u_j\}$ of group generators labeled $1...k$, where $(p_j)_{j=1}^k$ is a probability distribution on $1...k$ with at least some non-zero weights on non-identity unitaries. Let $\Phi^\lambda = \exp(- t \L)$ for $t \in \RR^+$. There exist states achieving arbitrarily large ratios
\[ I[A:B](\rho) / I[A:B]((\id^A \otimes \Phi^t)(\rho))\]
even for arbitrarily small $t > 0$.
\end{theorem}
Theorem \ref{thm:group} is proven in Subsection \ref{subsec:group}. As an application in Subsection \ref{sec:rate}, we show how flagged combinations of these channels can yield non-zero private capacity for channels that almost always send a nearly perfect copy of the input to their environment.

\section{Background}
We start by recalling existing continuity bounds on the relative entropy and related quantities. We recall the diamond norm given by
\[ \|\Phi\|_{\Diamond} = \sup_B \sup_{X : \|X\| = 1} \| (\Phi \otimes \Id^B)(\rho) \|_1 \]
for a completely positive, trace-perserving map (quantum channel) $\Phi$, where $B$ is an arbitrary finite-dimensional auxiliary system. We recall the quantum relative entropy $D(\rho \| \sigma) = \tr(\rho(\log \rho - \log \sigma))$ \cite{umegaki_conditional_1962}, where the logarithm is in most cases with respect to an arbitrary base. The quantum relative entropy generalizes the well-known classical Kullback-Leibler divergence. We recall a recent ``almost concavity" theorem for relative entropy that generalizes and unifies many known bounds on entropy, relative entropy, and mutual information of quantum and classical states:
\begin{theorem}[Theorem 5.1 from \cite{bluhm_continuity_2022}]\label{theo:theo_almost_concavity_relative_entropy}
    Let $(\rho_1, \sigma_1), (\rho_2, \sigma_2) \in \S_{\ker}$ with 
    \begin{equation}
        \S_{\ker} := \{(\rho, \sigma) \in \S(\Hil) \times \S(\Hil) \;:\; \ker \sigma \subseteq \ker \rho\}
    \end{equation} and $p \in [0, 1]$. Then, for $\rho = p \rho_1 + (1 - p) \rho_2$ and $\sigma = p \sigma_1 + (1 - p) \sigma_2$,
    \begin{equation}
        D(\rho \Vert \sigma) \ge p D(\rho_1 \Vert \sigma_1) + (1 - p) D(\rho_2 \Vert \sigma_2) - h(p)\frac{1}{2}\norm{\rho_1 - \rho_2}_1 - f_{c_1, c_2}(p) \, . 
    \end{equation}
    Here $f_{c_1,c_2}$ is defined in the original Theorem.
\end{theorem}
\noindent The binary entropy is given by
\[ h(p) := - p \log(p) - (1 - p) \log(1 - p) \pl, \]
for $p \in [0,1]$.
\noindent Rather than define $f_{c_1,c_2}$, we recall \cite[Proposition 5.2]{bluhm_continuity_2022} for a simplification:
\begin{cor} \label{cor:almostconcave}
Let $(\rho_1, \sigma_1), (\rho_2, \sigma_2) \in \S_{\ker}$ as in Theorem \ref{theo:theo_almost_concavity_relative_entropy}, $p \in [0,1]$, and $\rho = p \rho_1 + (1 - p) \rho_2$ and $\sigma = p \sigma_1 + (1 - p) \sigma_2$. Furthermore, assume that $\tilde{m} > 0$ lower bounds the non-zero eigenvalues of $\sigma_1 \oplus \sigma_2$. Then
\begin{equation}\label{eq:almost_concavity_relative_entropy}
        D(\rho \Vert \sigma) \ge p D(\rho_1 \Vert \sigma_1) + (1 - p) D(\rho_2 \Vert \sigma_2) f_{\tilde{m}}(p) 
    \end{equation}
where
\[ f_{\tilde{m}}(\epsilon) := h(\epsilon) 
    + \epsilon \log(\epsilon + (1 - \epsilon) \tilde{m}^{-1})
    + (1 - \epsilon) \log((1 - \epsilon) + \epsilon \tilde{m}^{-1}) \pl. \]
\end{cor}
 We recall Pinsker's inequality and a refinement \cite{canonne_short_2022} for the Kullback-Leibler divergence. 
\begin{prop}[Pinsker's Inequality \& Refinement] \label{prop:pinsker}
Let $\rho$ and $\sigma$ be density matrices on the same space. Then
\[ D(\rho \| \sigma) \geq \frac{1}{2} \|\rho - \sigma\|_1^2 \pl. \]
If $\rho$ and $\sigma$ commute, then
\[ D(\rho \| \sigma) \geq \max \Big \{ - \ln \Big (1 - \frac{1}{4} \|\rho - \sigma\|_1^2 \Big ), \frac{1}{2} \|\rho - \sigma\|_1^2 \Big \} \pl, \]
or equivalently,
\[ \| \rho - \sigma \|_1 \leq \sqrt{2 \min \big \{ D(\rho \| \sigma), 2 (1 - e^{- D(\rho \| \sigma)}) \big \}  } \]
\end{prop}

Bounds on relative entropy often follow from bounds on the semidefinite Loewner order of matrices. For matrices $\rho$ and $\sigma$, $c \rho \geq \sigma$ for scalar $c$ if $c \rho - \sigma \geq 0$, meaning it has all positive eigenvalues. There is a natural completely positive (cp) order on channels: $\Phi \geq_{cp} \Psi$ if
\[ (\Phi \otimes \Id^B)(\rho) \geq (\Psi \otimes \Id^B)(\rho) \]
for all extensions by a finite-dimensional auxiliary system $B$ and all densities $\rho$ on such extended systems. The symbols ``$\leq_{cp}, <_{cp},$ and $>_{cp}$ are defined correspondingly. The Pimsner-Popa index of a subspace (and more specifically, subalgebra) projection $\E$ is the minimum value of $c$ for which $c \E \geq_{cp} \Id$ \cite{pimsner_entropy_1986, gao_relative_2020}. We leverage some techniques from \cite{gao_complete_2022}. Recalling the integral Equations for relative entropy of \cite{gao_complete_2022} based on \cite{kastoryano_quantum_2013, carlen_gradient_2017, gao_fisher_2020},
\begin{equation} \label{eq:integralform}
    D(\rho \| \sigma)
    = \int_0^1 \int_0^s \| \rho_\zeta - \sigma \|_{(\rho, \sigma)_t^{-1}}^2 dt ds \pl,
\end{equation}
where
\[ \| \rho - \sigma \|_{\omega_t^{-1}}^2
    = \int_0^\infty (\rho - \sigma) \frac{1}{r + \omega} (\rho - \sigma) \frac{1}{r + \omega} d r \pl, \]
and
\[ (\rho, \sigma)_t := (1-t) \sigma + t \rho \pl. \]
We recall 
\begin{lemma}[Lemma 2.1 from \cite{gao_complete_2022}] \label{lem:normcomp}
If $\sigma \leq c \omega$, then $\| X \|_{\omega}^2 \leq c \| X \|_{\sigma}^2$ for any operator $X$ that is within the space of bounded operators for this norm.
\end{lemma}
\noindent We also recall \cite[Lemma 2.2]{gao_complete_2022}, that
\begin{equation} \label{eq:gaorouze}
    \rho \leq c \sigma \implies \kappa(c) \| \rho - \sigma \|_{\sigma_t^{-1}}^2 \leq D(\rho \| \sigma) \leq \| \rho - \sigma \|_{\sigma_t^{-1}}^2
\end{equation}
for
\begin{equation} \label{eq:kappa}
    \kappa(c) := (c \ln c - c + 1) / (c-1)^2 \pl.
\end{equation}
These techniques are generally used to transfer multiplicative comparisons from the cp-order of scaled densities to relative entropy.

\subsection{Semigroups}
We also use Lindbladians, although the paper is not specific to the traditional settings of Lindbladian dynamics. A Lindbladian $\L$ generates a family of quantum channels $\Phi^t = \exp(- \L t)$, known as a \textit{quantum Markov semigroup}. One may reparamterize many classes of known quantum channels to fit the Lindbladian formula. For example, let
\[ \tilde{\Phi}^{(\lambda)}(\rho) := (1 - \lambda) \rho + \lambda \E(\rho) \pl, \]
for all input densities $\rho$ and $\lambda \in [0,1]$. Then $\tilde{\Phi}^{(\lambda)} = \Phi^{g(\lambda)}$ for a semigroup $(\Phi^t)$, where $g : \RR^+ \rightarrow \RR^+$ is a continuous, non-decreasing function.

For a Lindbladian $\L$ generating semigroup $(\Phi^t)_{t \geq 0}$ with fixed point subspace projection $\E = \lim_{t \rightarrow \infty} \Phi^t$, a \textit{complete, modified logarithimc Sobolev inequality} (CMLSI) states that
\[ D((\Phi^t \otimes \Id^B)(\rho) \| (\E \otimes \Id^B)(\rho)) \leq e^{-\lambda t} D(\rho \| (\E \otimes \Id^B)(\rho)) \]
for all extensions by a finite-dimensional auxiliary system $B$ and all densities $\rho$ on such extended systems. As shown in \cite{gao_complete_2022} or via \cite{gao_geometric_2021} and \cite{junge_stability_2022}, CMLSI holds universally for semigroups with GNS detailed balance. The detailed balance condition is described in detail in \cite[Section 2.4]{gao_complete_2022}. In \cite{laracuente_self-restricting_2022}, CMLSI was shown to fail broadly for Lindbladians of the form $\L(\cdot) = i [\Ha, \cdot] + \S(\cdot)$, where $\S$ is a Lindbladian having fixed point projection $\E$, and time-evolution by $\Ha$ does not commute with $\E$. Nonetheless, the same work also showed multiplicative, exponential decay of relative entropy after fixing a discrete minimum timescale for semigroups that are self-adjoint with respect to the inner product $\langle X|Y \rangle = \tr(X^\dagger Y)$.
\begin{rem}
 Let $\Psi^{(t)}$ be a parameterized channel family such that $\Psi^{(0)} = \Id$, and $\Psi^t$ is a smooth function of $t$ on the interval $[0, s)$ for $s > 0$, or $s = \infty$. In many cases, one may define a semigroup $(\Phi^t)$ such that $\Psi^{(t)} = \Phi^{r(t)}$ for some smooth $r : [0,s) \rightarrow \RR^+$. For example, any channel of the form $\rho \mapsto (1-\epsilon) \rho + \epsilon \tilde{\Psi}(\rho)$ for a fixed channel $\tilde{\Psi}$ can be rewritten as the semigroup
 \[ \Phi^t(\rho) = e^{- r(t)} \rho  + \big (1 - e^{r(t)} \big ) \tilde{\Psi}(\rho) \pl. \]
 Well-known examples of this form include dephasing and depolarizing noise. We therefore emphasize that although quantum Markov semigroups are used throughout this paper to define the notion of arbitrarily small noise, the formulation generalizes naturally beyond the semigroup case and includes many commonly studied examples in quantum information theory.
\end{rem}

\section{Example(s) of Sudden Information Decay} \label{sec:sudden}
\begin{exam}[Depolarizing and Dephasing] \normalfont \label{exam:basic}
Consider for $\theta \in [0, \pi)$, $\lambda \in [0,1]$, and dimension $d$ the state
\begin{equation} \label{eq:rhothetalambda}
\begin{split}
\rho_{\theta, \lambda} & := (1-\lambda) \rho_{\theta,0} + \lambda \frac{\id}{d} \text{ , where } \\ 
\rho_{\theta, 0} & := \cos^2(\theta) \ketbra{0} + \cos \theta \sin \theta (\ket{0}\bra{1} + \ket{1}\bra{0})
    + \sin^2(\theta) \ketbra{1} \pl.
\end{split}
\end{equation}
We will observe sudden decay of the relative entropy $D(\rho_{\theta \lambda} \| \E_Z(\rho_{\theta, \lambda}))$, where $\E_Z$ is the pinching map to the computational basis. In any dimension, there are guaranteed to be at least three generalized Pauli bases, to which we refer as the $X,Y$, and $Z$ bases. When $d = 2$, $\rho_{\theta, \lambda}$ has the form of a state at angle $\theta$ from the $Z$-axis toward the $X$-axis on the Bloch sphere, then depolarized with strength $\lambda$. In this example, we will assume that the relative entropy is defined with respect to the natural logarithm. First,
\begin{equation} \label{eq:drtzero}
\begin{split}
    & D(\rho_{\theta, 0} \| \E_Z(\rho_{\theta, 0})) = H(\E_Z(\rho_{\theta, 0})) - H(\rho_{\theta, 0})
    \\ & = - \cos^2 \theta \ln \cos^2 \theta - \sin^2 \theta \ln \sin^2 \theta
    = \Omega \Big (\theta^2 \ln \Big (\frac{1}{\theta} \Big ) \Big ) \pl.
\end{split}
\end{equation}
Here we use the notation $\ln(")$ to denote the natural logarithm of whatever multiplies it to the left. For the subtracted entropy term, we can use the fact that $\rho_{\theta, \lambda} = (1-\lambda) \rho_{\theta, 0} + \lambda \id/d$, where $\rho_{\theta,0}$ is a pure state. Hence
\begin{equation} \label{eq:Hrtl}
    H(\rho_{\theta, \lambda}) = - (d-1) \frac{\lambda}{d} \ln \frac{\lambda}{d}
        - \Big ( 1 - \lambda \frac{d-1}{d} \Big ) \ln \Big ( 1 - \lambda \frac{d-1}{d} \Big ) \pl.
\end{equation}
Now consider the case that $\lambda > 0$, and $\theta << \lambda$. We calculate
\begin{equation} \label{eq:Hrtltheta}
\begin{split}
    H(\E_Z(\rho_{\theta, \lambda})) = 
         \big ((1-\lambda) \cos^2(\theta) + \frac{\lambda}{d} \big ) \ln ("_1)
         + \big ((1-\lambda) \sin^2(\theta) + \frac{\lambda}{d} \big ) \ln("_2) \pl.
\end{split}
\end{equation}
To simplify in leading order,
\[ \big ((1-\lambda) \cos^2(\theta) + \frac{\lambda}{d} \big ) \approx 
    \Big (1-\lambda \frac{d-1}{d} - \frac{\theta^2}{2} (1 - \lambda)  \Big ) \]
up to a correction of $O(\theta^4)$. Hence
\[ \ln("_1) \approx \ln \Big ( 1-\lambda \frac{d-1}{d} \Big )- \frac{1 - \lambda}{1-\lambda (d-1)/d} \frac{\theta^2}{2} \pl. \]
Expanding to leading order,
\begin{equation} \label{eq:Hzrtl1}
\big ((1-\lambda) \cos^2(\theta) + \frac{\lambda}{d} \big ) \ln("_1)
    \approx \Big (1 - \lambda \frac{d-1}{d} \Big ) \ln \Big (1 - \lambda \frac{d-1}{d}  \Big )
     - \frac{\theta^2}{2} \Big ((1-\lambda) \Big ( \ln \Big (1 - \lambda \frac{d-1}{d}  \Big ) - 1 \Big ) \Big ) \pl.
\end{equation}
For the other term,
\[ \big ((1-\lambda) \sin^2(\theta) + \frac{\lambda}{d} \big )
    \approx \frac{\lambda}{d} + \theta^2 (1 - \lambda) \text{, and }
    \ln("_2) \approx \ln \frac{\lambda}{d} + \theta^2 d \Big (\frac{1}{\lambda} - 1 \Big ) \pl. \]
up to $O(\theta^3)$. Hence 
\begin{equation} \label{eq:Hzrtl2}
    \big ((1-\lambda) \sin^2(\theta) + \frac{\lambda}{2} \big ) \ln("_2)
    \approx \frac{\lambda}{d} \ln \frac{\lambda}{d}
    + \theta^2 \Big ((1 - \lambda) \ln \frac{\lambda}{d} + 1 - \lambda \Big ) \pl.
\end{equation}
Combining Equations \eqref{eq:Hrtl}, \eqref{eq:Hzrtl1}, and \eqref{eq:Hzrtl2},
\[ D(\rho_{\theta, \lambda} \| \E_Z(\rho_{\theta, \lambda})) = 
    \theta^2 \Big ( \frac{1}{2} \Big ( (1-\lambda) \ln \Big (1 - \lambda \frac{d-1}{d} \Big ) + 1 - \lambda \Big )
    - (1 - \lambda) \ln \frac{\lambda}{d} \Big ) + O(\theta^3 \ln (1/\theta)) \pl. \]
Comparing with Equation \eqref{eq:drtzero} for small $\theta$,
\[ \frac{D(\rho_{\theta, \lambda} \| \E_Z(\rho_{\theta, \lambda}))}{D(\rho_{\theta, 0} \| \E_Z(\rho_{\theta, 0}))}
    = O(1 / \ln(1 / \theta)) \pl. \]
We also observe fragility of quantum mutual information. Let
\[ \omega_{\theta, \lambda} := \frac{1}{2} \ketbra{0}^A \otimes \rho_{\theta, \lambda}^B + \frac{1}{2} \ketbra{1}^A \otimes \rho_{-\theta, \lambda}^B \pl. \]
This state follows from depolarizing the $B$ system of $\omega_{\theta, 0}$ with strength $\lambda$. We note that
\[ \omega_{\theta, \lambda}^B = \tr_A(\omega_{\theta, \lambda}) = \E_Z(\rho_{\theta, \lambda}) \pl. \]
Hence $H[B](\omega_{\theta, \lambda}) = H(\E_Z(\rho_{\theta, \lambda}))$. Furthermore, $H[B|A](\omega_ {\theta, \lambda}) = H(\rho_{\theta, \lambda})$. Hence
\begin{equation} \label{eq:mutinfo1}
\ I[A:B](\omega_{\theta, \lambda}) = H[B](\omega_{\theta, \lambda}) - H[B|A](\omega_{\theta, \lambda})
    = D(\rho_{\theta, \lambda} \| \E_Z(\rho_{\theta, \lambda})) \pl.
\end{equation}
The same sudden decay thereby holds for mutual information as did for relative entropy with respect to $\E_Z$.

Finally, we observe that $\E_Y \E_Z(\cdot) = \E_Z \E_Y(\cdot) = \hat{1}/d$, and that
\[ \rho_{\theta, \lambda} = (1-\lambda) (\rho_{\theta, 0}) + \lambda \E_Y (\rho_{\theta, 0}) \pl. \]
Therefore, we may substitute $Y$ basis dephasing for depolarizing noise, observing the sudden decay of relative entropy and mutual information for dephasing noise as well.
\end{exam}

\begin{figure}[h!] \scriptsize \centering
	\begin{subfigure}[b]{0.40\textwidth}
		\includegraphics[width=0.98\textwidth]{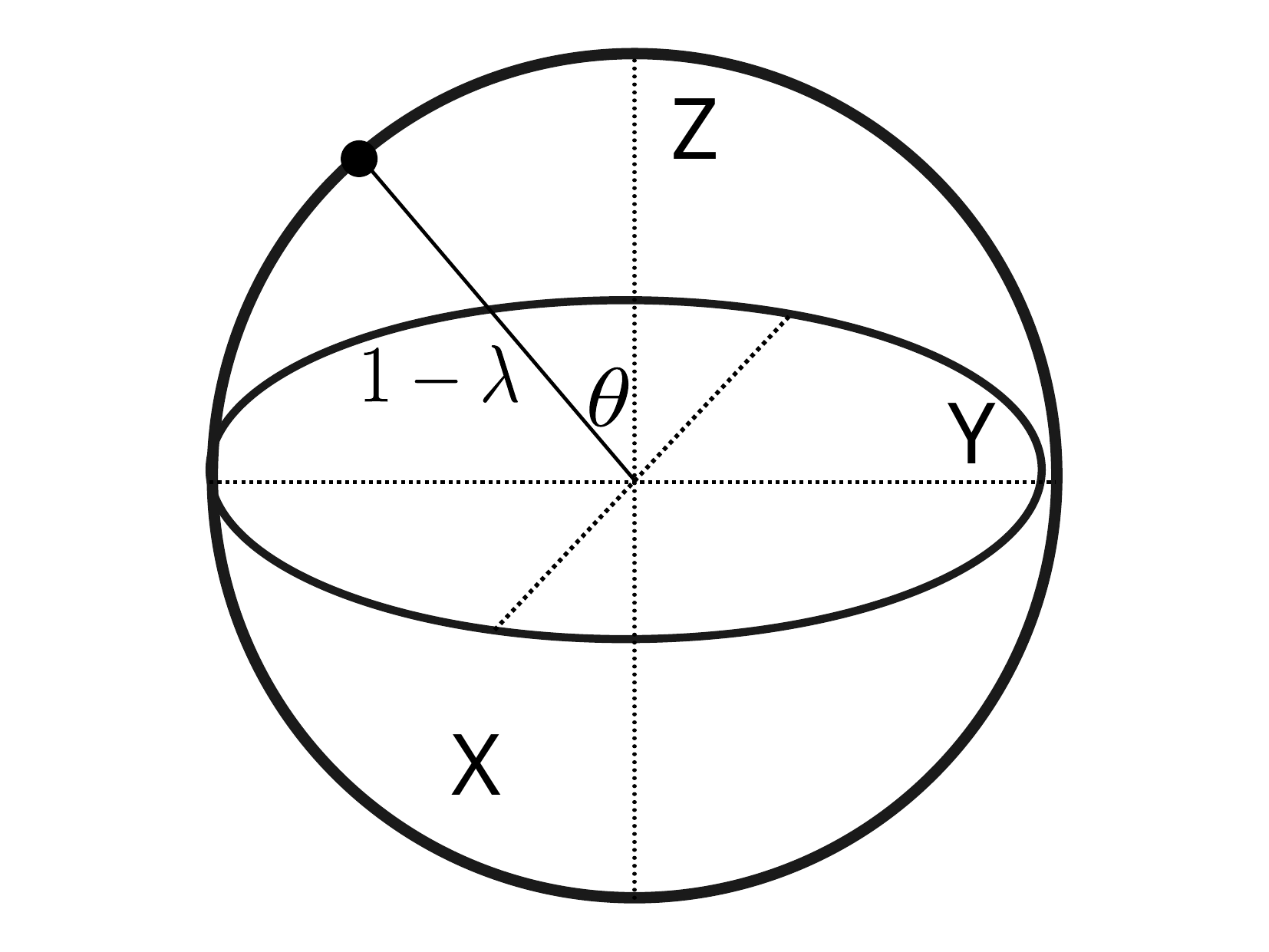}
		\caption{}
	\end{subfigure}
	\begin{subfigure}[b]{0.40\textwidth}
		\includegraphics[width=0.98\textwidth]{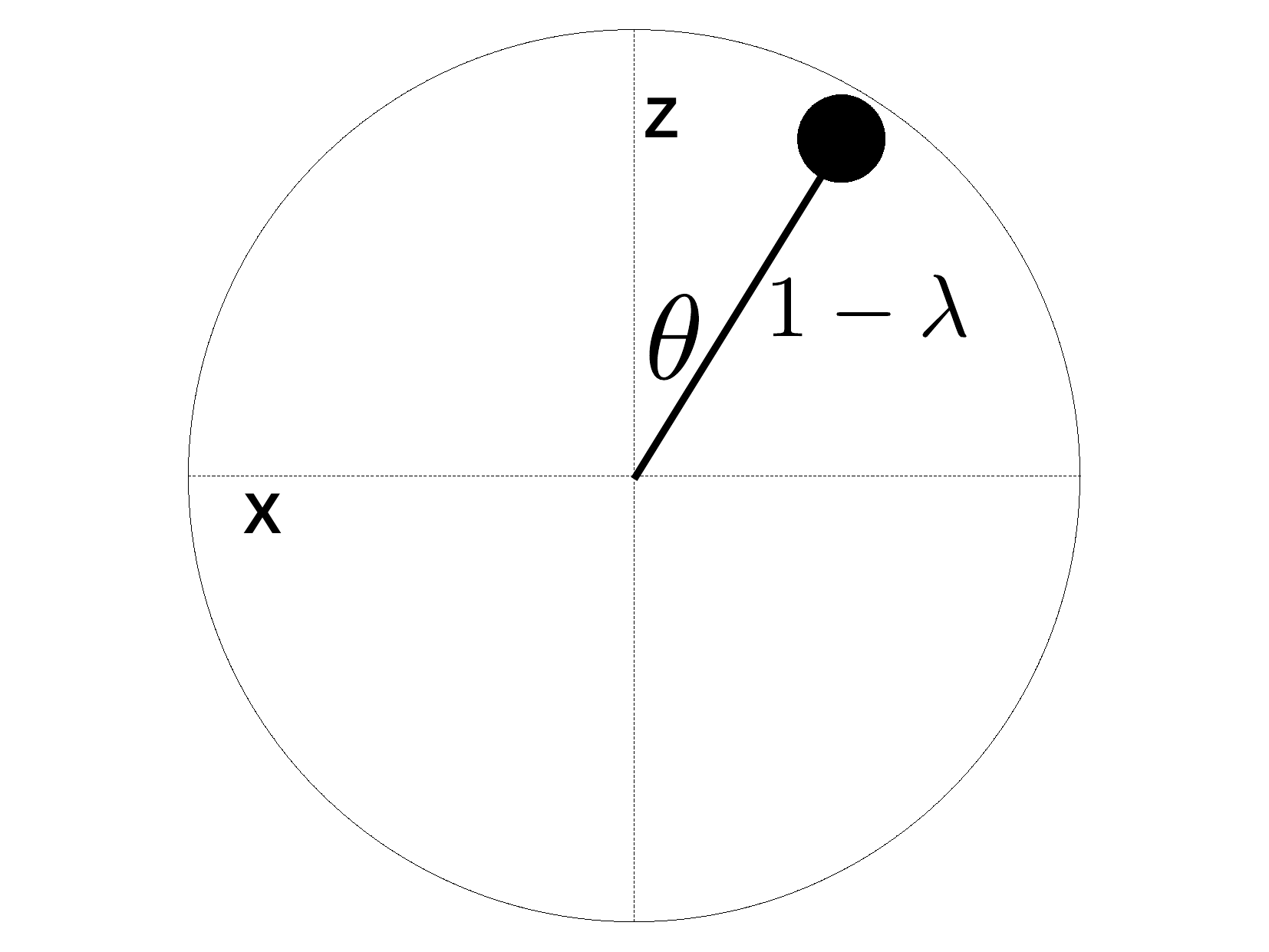}
		\caption{}
	\end{subfigure}
	\caption{(A) Bloch sphere visualization of $\rho_{\theta, \lambda}$ when $d=2$. (B) Visualization of the X-Z plane in the Bloch sphere. \label{fig:visual} }
\end{figure}

\subsection{Proof of Fragility Under Finite Group Channels} \label{subsec:group}
After seeing Example \ref{exam:basic}, one may wonder how generally the phenomenon applies. We conjecture that fragility of mutual information is broad, and to illustrate this, we show that it holds for a wide range of channels given by quantum group generators. In particular, this set of channels includes the Pauli channels on multi-qubit systems, channels representing finite graphs as studied in \cite{laracuente_quasi-factorization_2022}, channels built from Schur multipliers with conjugation symmetry, and a wide range of other channels. See \cite{bardet_group_2021} for more examples.
\begin{lemma} \label{lem:convexreplace}
 Let $\E$ be a unital, completely positive, finite-dimensional, trace-preserving conditional expectation. Let
 \[ \Phi^t(\rho) =  e^{- t} \rho + \big (1 - e^{- t} \big ) \E(\rho) \pl. \]
 Then there exist states achieving arbitrarily large ratios
\[ I[A:B](\rho) / I[A:B]((\id^A \otimes \Phi^t)(\rho))\]
even for arbitrarily small $t > 0$.
\end{lemma}
\begin{proof}
Recall the block diagonal form of a conditional expectation as first discovered by von Neumman \cite{von_neumann_rings_1949} and reviewed for finite-dimensional channels in \cite{gao_unifying_2017}. In our case:
\begin{equation} \label{eq:blockcondexp}
\E(\rho) = \oplus_l \tr_{C_l}(P_l \rho P_l) \otimes \id^{C_l} / |C_l| \pl,
\end{equation}
where $P_l$ is a projection to the $l$th diagonal block. The conditional expectation first reduces $\rho$ to $\oplus_l P_l \rho P_l$, a block diagonal matrix with entries from $\rho$, effectively removing all coherence between blocks. We may subsequently interpret each such block as a bipartite system $B_l \otimes C_l$, in which $\E$ completely depolarizes subsystem $C_l$.

Since we use only 2-dimensional subspace of $A$, and since the channels acting on $B$ cannot change that, we are free to assume that $|A| = 2$ without loss of generality.

First, we assume that at least one block of $\E$ as in Equation \eqref{eq:blockcondexp} has a depolarized subsystem, which we here denote $C_0$. We may then complete the Lemma by inputting the the state
\[ \tilde{\omega}_{\theta, \lambda} := \frac{1}{2} \ketbra{0}^A \otimes \ketbra{0}^{B_0} \otimes \rho_{\theta, \lambda}^{C_0} + \frac{1}{2} \ketbra{1}^A \otimes \ketbra{0}^{B_0} \otimes \rho_{-\theta, \lambda}^{C_0} \pl, \]
 where the computational basis $\ket{0}$ and $\ket{1}$ states are any two orthogonal pure states on $C_0$, and $\rho_{\theta,\lambda}$ is as in Examples \ref{exam:basic}. We note the resemblance to $\omega_{\theta, \lambda}$ as in Example \ref{exam:basic}. In this case, the proof follows from the same calculation as in the example, as the subsystem $B_0$ contributes nothing, and as
 \[ I[A:B](\tilde{\omega}_{\theta, \lambda}) = D(\rho^{C_0}_{\theta,\lambda} \| \E_Z(\rho^{C_0}_{\theta,\lambda})) \pl. \]

 For the remaining case, we assume that $\E$ does not depolarize any blocks. We again follow Example \ref{exam:basic}. For $\E$ to be non-trivial, there must then be at least two blocks. We take a basis in which $\E$ is block diagonal as the ``$Y$" basis, constructing $\rho_{\theta,\lambda}$ with respect to two mutually unbiased bases, ``$X$" and ``$Z$." Then
 \[ \E(\rho_{\theta,\lambda}) = \frac{\id}{2} \otimes \frac{1}{d} \sum_{j=1}^d \ketbra{j}^B \]
 for some $d \leq |B|$. Hence $\E$ acts on $\rho_{\theta,\lambda}$ as depolarizing noise, potentially on a subsystem of $B$. This observation reduces the current case to that in which $\E$ depolarizes some blocks, for which the result is already shown above.

\end{proof}
\begin{proof}[Proof of Theorem \ref{thm:group}]
Via elementary identities,
\[ \Phi^t(\rho) = \big ( \Phi^{t/m}(\rho) \big )^m = \big ( e^{-t/m}(\rho) + \epsilon_{t/m} \sum_{j = 1}^k \big ( \frac{1}{2} p_j u_j \rho u_j^\dagger + \frac{1}{2}  p_j u_j^\dagger \rho u_j \big ) + \big ( 1 - e^{-t/m} - \epsilon_{t/m} \big ) \Psi_{t/m}(\rho))^m \]
for sufficiently large $m$, some channel $\Psi_{t/m}$, and some $\epsilon_{t/m} \in (0,1)$. Hence
\[ \Phi^t(\rho) \geq \epsilon_{t/m} \sum_{j_1, ..., j_m} \Big ( \frac{1}{2} \prod_{a=1}^n p_1 ... p_m u_1 ... u_m \rho u_m^\dagger ... u_1^\dagger + \frac{1}{2} \prod_{a=1}^m p_1 ... p_n u_1 ... u_m \rho u_m^\dagger ... u_1^\dagger \Big ) \pl. \]
Sufficiently long products of generators eventually generate all unitaries in the group. For any $t > 0$ and sufficiently large $m$, $\exists$ an $\eta(t, \L) > 0$ such that $\Phi^t \geq_{cp} \eta(t, \L) \E$. Here $\E$ is the fixed point projection corresponding to the group's Haar measure. Hence
\[ \Phi^t = \eta(t,\L) \E + (1 - \eta(t, \L)) \Psi \]
for some channel $\Psi$, which also has $\E$ as a fixed point projection.

Since mutual information is non-increasing under local operations, we may rewrite
\[ \Phi^t = \eta(t,\L) \E + (1 - \eta(t, \L)) \Psi = \Psi(\eta(t,\L) \E + (1 - \eta(t, \L)) \text{Id}) \pl, \]
then note that $\Phi^t$ induces at least as much decay as would
\[ \tilde{\Phi}^t := \eta(t,\L) \E + (1 - \eta(t, \L)) \text{Id} \pl. \]
Hence it suffices to show fragility under convex replacement by a conditional expectation. The Theorem then follows Lemma \ref{lem:convexreplace}.
\end{proof}

\subsection{Channels with Non-Zero Private Rate Despite Arbitrarily High Loss} \label{sec:rate}
Consider the channel given by
\[ \Psi_{\lambda, p}(\rho) = p \ketbra{0} \otimes \rho + (1 - p) \ketbra{1} \otimes \Phi^c_\lambda(\rho) \pl, \]
where $\Phi_\lambda$ is an instance of a quantum channel family with noise parameter $\lambda$. For the remainder of this Subsection, let $\bar{\Psi}$ denote the complementary channel of channel $\Phi$. We will assume that $\lambda$ is small, but $p$ is also small. Hence most channel instances leak an almost perfect copy of the input to the environment, while a few retain a perfect copy at the output. The auxiliary flag state for $\Psi_{\epsilon, \lambda, p}$ is copied in both output and environment.

If Alice has full control over the input state, and $\Phi_\lambda$ exhibits non-uniform mutual information decay, then Alice can achieve some non-zero private classical transmission rate for arbitrarily small $p$ and $\lambda$, despite the fact that the environment would normally recover most of the output information. We model private information transmission via a channel $\Psi : A' \rightarrow B$ tensored with the identity on a reference system $A$, yielding a state on bipartite system $AB$. The complement, $\bar{\Psi}$, outputs to environment system $E$. Recall that the private capacity of a quantum channel \cite{devetak_private_2005} $\Psi$ is given by
\[ C_p(\Psi) := \lim_{n \rightarrow \infty} \frac{1}{n} \sup_\rho \{ I[A:B](\Psi^{\otimes n}(\rho)) - I[A:E](\bar{\Psi}^{\otimes n}(\rho)) \} \pl, \]
where $\rho$ is on $n$ copies of the original input and reference systems.  By superadditivity of the private information rate under tensor product and additivity of the mutual information under convex combination of orthogonal states,
\[ C_p(\Psi_{\lambda, p}) \geq p I(A:B)_{\rho} - (1-p) I(A:E)_{\Phi_{\lambda}(\rho)} \]
for any choice of input state $\rho$. If we may choose $\rho$ such that $I(A:B)_{\rho} / I(A:E)_{\Phi_{\lambda}(\rho)}$ is arbitrarily large, then $C_p(\Psi_{\lambda, p})$ is positive for any $p > 0$ and $\lambda > 0$.

\begin{exam} \normalfont
Recalling Example \ref{exam:basic}, we demonstrate the result more concretely. In this example, let $\Phi^\lambda$ be a qubit depolarizing channel. The input states are $\rho_\theta$ as in Example \ref{exam:basic}. For any $p$, we simply need $\theta$ sufficiently small.

As noted in Example \ref{exam:basic}, we may also take $\Phi^\lambda$ to be a qubit dephasing channel in the Pauli $Y$ basis. In this case, the channel constructed is equivalent up to local rotations to the complement of the dephrasure channel, for which a similar phenomenon was observed in \cite[Appendix H, ArXiv Version]{leditzky_dephrasure_2018}. Our results generalize and explain this appearance of non-zero private capacity for channels with high loss to the environment.
\end{exam}

\section{Situation-Specific Decay Converse Bounds}
\subsection{CMLSI Converse Bounds} Decay via a semigroup toward a fixed point projection of that state is a special case of broader decay with additional structure. In this Subsection, we exploit that additional structure to upper bound decay rate uniformly in the state. To simplify notation, let
\begin{equation}
\epow{a}{m} := \frac{a^m \exp(a)}{m!}
\end{equation}
for any scalar $a > 0$, $k \in \NN$.
\begin{lemma} \label{lem:ereplace}
For any $t > 0$, input density $\rho$, and Lindbladian $\S$, let
\begin{equation} \label{eq:taylortail}
f(\S, t)(\rho) := \rho - e^{-t \S}(\rho) = t \S(\rho) - \frac{t^2 \S^2}{2}(\rho) + ...
\end{equation}
Let $\E$ be any projection such that for given $c > 0$, $c \E \geq_{cp} \exp(- t \S)$ for any $t \geq 0$, and $\E \circ \exp(-t \S) = \E$. Then the map given by
\[ \sigma \mapsto \E(\sigma) {}^+_- 2 f(\S, t)(\sigma) / c \epow{t \|\S\|_{\Diamond}}{1} \]
on densities is a quantum channel for both the ``$+$" and the ``$-$" case.
\end{lemma}
\begin{proof}
Since $f(\S, t)(\rho)$ is a Hermitian matrix with real spectrum, it has a decomposition into positive and negative parts $f_+ > 0$ and $ f_- > 0$ such that $f(\S, t)(\rho) = f_+ - f_-$, and $\text{supp}(f_+) \perp \text{supp}(f_-)$. Furthermore, since $\exp(- t \S)(\rho) \geq 0$, $\rho \geq f(\S, t)(\rho)$.

Now fix $\rho = \frac{1}{\sqrt{d}} \sum_{i=1}^d \ket{i} \otimes \ket{i}$, denoting the canonical Bell state in the computational basis. Now since $\rho \geq f_+$, and $\rho$ is rank 1, $f_+$ is also rank 1, and $f_+ = g(\S, t) \rho$ for some $g : (L(\MM_d, \MM_d), \RR^+) \rightarrow \RR^+$. Because $\tr(f(\S, t)) = 0$, and $f_+, f_-$ are positive, $\|f_+\|_\infty = \|f_+\|_1 = \|f_-\|_1$, and $\|f_+\|_1 + \|f_-\|_1 = \|f(\S, t)\|_1$. Hence $\|f_+\|_1 = \|f(\S, t)\|_1 / 2$. Via the triangle inequality and submultiplicativity of the diamond norm under composition, using as well that $\rho$ is rank 1, $g(\S, t) \leq \|f(\S, t)\|_1 \leq \epow{t \|\S\|_{\Diamond}}{1}$.

Knowing that $\rho \leq c \E(\rho)$, $f_+ \leq c g(\S, t) \E(\rho)$. Since $\E(\rho) = \E(\exp(- t \S)(\rho))$, $\E(f(\S, t)) = 0$. Since $\E$ is completely positive, $\E(f_+(\rho)) = \E(f_-(\rho))$.  Since $f_- \geq 0$, $f_- \leq c \E(f_-) = c \E(f_+) = c \E(g(\S, t) \rho) = c g(\S, t) \E(\rho)$. We thereby know that $c \tilde{g}(\S, t) \E(\rho) {}^+_- f(\S, t) \geq 0$ for any $\tilde{g}(\S t) \leq g(\S, t)$. We recall that $f(\S, t)$ has trace 0 and that $\E(\rho)$ has trace 1, so $\E(\rho) - f(\S, t) / c g(\S, t)$ has trace 1. It is a valid density matrix and the output of a linear map on a Bell pair. Via the Choi-Jamiolkowski isomorphism, a linear map that outputs a density on the Bell pair input has a valid Choi matrix, so it is a quantum channel (completely positive, trace-preserving map).
\end{proof}
\begin{lemma} \label{lem:scombine}
Let $(\Phi^t : t \in \RR^+)$ be a continuous quantum Markov semigroup generated by $\S$, and $\E$ project to a subspace that is fixed under $\Phi^t$ such that $c \E \geq_{cp} \text{Id}$. Let $\rho$ and $\eta$ be any input densities. Then for any $t > 0$,
\[ D(\Phi^t(\rho) \| \Phi^t(\eta)) \geq
    D(e^{- t c \|\S\|_{\Diamond} \S_{\text{rep}} / 2}(\rho) \| e^{- t c \|\S\|_{\Diamond} \S_{\text{rep}} / 2}(\eta)) \pl, \]
where $ \S_{\text{rep}}(\sigma) := \sigma - \E(\sigma)$.
\end{lemma}
\begin{proof}
Let $\tau > 0$, and $f(\S, \tau)$ be defined as in Equation \eqref{eq:taylortail}. Let
\[ \Psi(\rho) := \E(\rho) - 2 f(\S, \tau)(\rho) / c \epow{\tau \|\S\|_{\Diamond}}{1} \]
Via Lemma \ref{lem:ereplace}, $\Psi(\rho)$ is the output of a quantum channel on a density, so it is a density. Hence the map $\tilde{\S}$ given by
\[ \tilde{\S}(\sigma) := \sigma - \Psi(\sigma) \]
is a valid Lindbladian, which replaces $\sigma$ by $\Psi(\sigma)$ over time. We also note that by its explicit form,
\begin{equation} \label{eq:scombine}
(c \epow{\tau \|\S\|_{\Diamond}}{1} \tilde{\S} / 2 + \tau \S )(\sigma) = (\sigma - \E(\sigma)) \tau c \|\S\|_{\Diamond} / 2 + O(\tau^2) \pl.
\end{equation}
We observe that this map is the Lindbladian generator for a semigroup that commutes with $\E$ and has $\E$ as its fixed point projection. Via Equation \eqref{eq:scombine}, the data processing inequality, and continuity of relative entropy with respect to a conditional expectation (see \cite[Lemma 7]{winter_tight_2016})
for any $\tau > 0$,
\begin{equation*}
\begin{split}
& D(\Phi^t(\rho) \|  \E(\Phi^t(\rho)))  
    \geq D(e^{- \tau c \|\S\|_{\Diamond} \tilde{S}} \Phi^{t}(\rho)
        \| e^{-\tau c \|\S\|_{\Diamond} \tilde{S} / 2} (\Phi^t(\eta))) \pl.
\\ & \geq D(e^{- \tau c \|\S\|_{\Diamond} \S_{\text{rep}} / 2} \Phi^{t-\tau}(\rho)
    \| e^{-\tau c \|\S\|_{\Diamond} \S_{\text{rep}} / 2} (\Phi^{t-\tau}(\eta))) - O(\tau^2 \ln (1/\tau))
\\ & = D(e^{- \tau c \|\S\|_{\Diamond} / 2} \Phi^{t-\tau}(\rho) + (1 - e^{- \tau c \|\S\|_{\Diamond} / 2}) \E(\rho)
     \pl \pl \| \\ & \hspace{0.2\textwidth}
      e^{- \tau c \|\S\|_{\Diamond} / 2} \Phi^{t-\tau}(\eta) + (1 - e^{- \tau c \|\S\|_{\Diamond} / 2}) \E(\eta))
    - O(\tau^2 \ln (1/\tau))
\\ & = D(\Phi^{t-\tau}(e^{- \tau c \|\S\|_{\Diamond} / 2} \rho + (1 - e^{- \tau c \|\S\|_{\Diamond} / 2}) \E(\rho))
     \pl \pl \| \\ & \hspace{0.2\textwidth}
     \Phi^{t-\tau}(e^{- \tau c \|\S\|_{\Diamond} / 2} \eta + (1 - e^{- \tau c \|\S\|_{\Diamond} / 2}) \E(\eta)))
    - O(\tau^2 \ln (1/\tau)) \pl.
\end{split}
\end{equation*}
Iterating the inequality $t/\tau$ times as $\tau \rightarrow 0$, we obtain that
\[ D(\Phi^t(\rho) \| \Phi^t(\eta))
    \geq D(e^{- t c \|\S\|_{\Diamond} \S_{\text{rep}} / 2}(\rho) \| e^{- t c \|\S\|_{\Diamond} \S_{\text{rep}} / 2}(\eta)) \pl. \]
To complete the Lemma, we again apply the chain rule for relative entropy.
\end{proof}
In MLSI and related notions, we study $D(\Phi^t(\rho) \| \sigma)$, where $\sigma$ is a fixed point of $\Phi^t$ such that $\Phi^\infty(\rho) = \sigma$. In these cases, general converse bounds follow, as decay toward $\sigma$ yields a state pair that is order-comparable to the original.
\begin{lemma} \label{lem:replacementconverse}
For densities $\rho$ and $\sigma$ such that $\rho \leq c \sigma$,
\[ D((1 - \zeta) \rho + \zeta \sigma \| \sigma) \geq \sup_{\tau \in (0,1)} \frac{(1 - \zeta)^2 \tau}{\tau + \zeta}
    \Big (1 - \frac{\tau(1 - \ln \tau)}{\kappa(c)} \Big) D(\rho \| \sigma) \]
for any $\tau \in (0,1)$, and $\kappa(c)$ as in Equation \eqref{eq:kappa}.
\end{lemma}
\begin{proof}
Let $\rho_\zeta = (1 - \zeta) \rho + \zeta \sigma$. For any $t \in (0,1)$,
\[ (\rho_\zeta, \sigma)_t \leq \frac{1 - (1 - \zeta) t}{1-t} (\rho, \sigma)_t \pl. \]
Hence via Lemma \ref{lem:normcomp} for any $X$,\begin{equation}
\begin{split}
\| X \|_{(\rho, \sigma)_t^{-1}}^2
    & \leq \Big (1 + \frac{\zeta}{1-t} \Big ) \| X \|_{(\rho_\zeta, \sigma)_t^{-1}}^2  \pl.
\end{split}
\end{equation}
To overestimate $D(\rho \| \sigma)$ in terms of $D(\rho_\zeta \| \sigma)$, we would then use Equation \eqref{eq:integralform}, implying that
\[ D(\rho \| \sigma) = \int_0^{1} \int_0^s \| \rho - \sigma \|_{(\rho, \sigma)_t^{-1}}^2 dt ds 
    \leq \int_0^{1} \int_0^s \Big (1 + \frac{\zeta}{1-t} \Big ) \| \rho_\zeta - \sigma \|_{(\rho_\zeta, \sigma)_t^{-1}}^2  dt ds \pl. \] 
Performing this integral in general is difficult, because the norm is time-dependent in a way that also depends on the densities involved. One can use the simple estimate that for arbitrary $X$ and $\tau \in (0,1)$,
\begin{equation} \label{eq:intpart1}
\begin{split}
\int_0^{1-\tau} \int_0^s \| X \|_{(\rho, \sigma)_t^{-1}}^2 dt ds \leq \Big (1 + \frac{\zeta}{\tau} \Big ) \int_0^{1-\tau} \int_0^s \| X \|_{(\rho_\zeta, \sigma)_t^{-1}}^2 dt ds \pl.
\end{split}
\end{equation}
To use this estimate, we must show that we can estimate the $\tau = 0$ case by some for which $\tau > 0$, avoiding the zero denominator. Note that
\begin{equation} \label{eq:zcomp1}(
1 + (1 - \zeta) (c-1) t ) \sigma \geq (\rho_\zeta, \sigma)_t \geq (1 - (1 - \zeta) t) \sigma
\end{equation}
Via the second inequality at $\zeta = 0$ and Lemma \ref{lem:normcomp}, $\|X\|_{(\rho_\zeta, \sigma)_t^{-1}}^2 \leq \|X\|_{\sigma^{-1}}^2 / (1 - t)$ for any $X$. We calculate that
\begin{equation*}
\begin{split}
& \int_{1-\tau}^1 \int_0^s \frac{1}{1 - t} dt ds = - (1-\zeta) \int_{1-\tau}^1 \ln(1-s) ds
    = \tau(1 - \ln \tau) \pl.
\end{split}
\end{equation*}
Since $\|X\|_{\sigma^{-1}}$ has no $t$-dependence, we have via Equation \eqref{eq:gaorouze} that
\begin{equation*}
\begin{split}
    & \int_{1-\tau}^1 \int_0^s \|\rho - \sigma \|_{(\rho, \sigma)_t^{-1}}^2 dt ds
        \leq \tau(1 - \ln \tau) \|\rho - \sigma \|_{\sigma^{-1}}^2
\\ \leq &  \frac{\tau(1 - \ln \tau)}{\kappa(c)} D(\rho \| \sigma)
    = \frac{\tau(1 - \ln \tau)}{\kappa(c)} \int_0^1 \int_0^s \|\rho - \sigma \|_{(\rho, \sigma)_t^{-1}}^2 dt ds \pl.
\end{split}
\end{equation*}
We also find that
\begin{equation} \label{eq:intpart3}
\int_{0}^{1 - \tau} \int_0^s \|\rho - \sigma \|_{(\rho, \sigma)_t^{-1}}^2 dt ds
    \geq \Big (1 - \frac{\tau(1 - \ln \tau)}{\kappa(c)} \Big)
    \int_0^1 \int_0^s \|\rho - \sigma \|_{(\rho, \sigma)_t^{-1}}^2 dt ds \pl.
\end{equation}
Finally, we observe that $\| \rho_\zeta - \sigma \| = (1 - \zeta) \| \rho - \sigma \|$ for any norm. Combining this observation with Equations \eqref{eq:intpart1} and \eqref{eq:intpart3} completes the Lemma.
\end{proof}
\begin{theorem} \label{thm:clsiconverse}
Let $\Phi^t = \exp( - t \L)$ be a quantum Markov semigroup generated by Lindbladian $\L$. Assume that $\E$ is a projection to a fixed point of $\L$, and $c \E \geq_{cp} \Id$. Then
\[ D(\Phi^t(\rho) \| \E(\rho)) \geq g \big ( 1 - e^{- t c \|\L\|_{\Diamond}}, c \big ) D(\rho \| \E(\rho)) \pl \]
with
\[ g(\zeta, c) := \sup_{\tau \in (0,1)} \frac{(1 - \zeta)^2 \tau}{\tau + \zeta} \Big (1 - \frac{\tau(1 - \ln \tau)}{\kappa(c)} \Big) \pl,\]
and $\kappa(c)$ as in Equation \eqref{eq:kappa}. The function $g(\zeta, c) \rightarrow 1$ as $\zeta \rightarrow 0$, and $g(\zeta,c) > 0$ for all $\zeta < 1$ and all $c$.
\end{theorem}
\begin{proof}
Via Lemma \ref{lem:scombine},
\[ D(\Phi^t(\rho) \| \E(\rho)) \geq D(e^{- t c \|\L\|_{\Diamond} \S_{rep} / 2} (\rho) \| \E(\rho)) \pl, \]
where $\S_{rep}(\rho) := \rho - \E(\rho)$. Hence
\[ e^{- t c \|\L\|_{\Diamond} \S_{rep} / 2}(\rho)
    = \big ( 1 - e^{- t c \|\L\|_{\Diamond}} \big ) \E(\rho) + e^{- t c \|\L\|_{\Diamond}} \rho \pl. \]
We then apply Lemma \ref{lem:replacementconverse}.

To see that $g(\zeta, c) \rightarrow 1$ as $\zeta \rightarrow 0$, one may substitute $\tau = \sqrt{\zeta}$, then immediately observe the limiting behavior of the expression. To see that $g(\zeta, c)$ is always positive, note that $\tau (1 - \ln \tau)$ can be made arbitrarily small for sufficiently small $\tau$.
\end{proof}
\begin{rem} \label{rem:completeclsicon}
    In Theorem \ref{thm:clsiconverse}, both $\|\L\|_{\Diamond}$ and $c$ are stable under the extension $\Phi^t \rightarrow \Phi^t \otimes \Id^B$. Therefore, the inequality trivially extends to
    \[ D((\Phi^t \otimes \Id^B)(\rho) \| (\E\otimes \Id^B)(\rho)) \geq g \big ( 1 - e^{- t c \|\L\|_{\Diamond}}, c \big ) D(\rho \| (\E \otimes \Id^B)(\rho)) \]
    for any extension by a finite-dimensional auxiliary system $B$, and any $\rho$ on the joint system with this extension. Like CMLSI, this converse is also ``complete" in the sense of extending to include auxiliaries, and tensor-stable.
\end{rem}

\begin{exam}[Qubit Depolarizing Noise] \normalfont
    To illustrate Theorem \ref{thm:clsiconverse}, we consider the basic case of depolarizing noise via the channel
    \[ \Phi^t(\rho) = e^{- t} \rho + \big (1 - e^{-t} \big ) \frac{\id}{2} \]
    on system $A$ of dimension $2$. This channel has 1-CMLSI, so $D((\Phi^t \otimes \Id^B)(\rho) \| \id / 2 \otimes \rho^B) \leq e^{-t} D(\rho \| \id/2 \otimes \rho^B)$ for every input density $\rho$ on every possible extension $B$. Here we find that $c$ as in Theorem \ref{thm:clsiconverse} is bounded by $c \leq 4$, and that $\|\L\|_{\Diamond} \leq (2^2-1)/4 = 3/4$. Hence for $d=2$, $\zeta = 1 - \exp(- 3 t)$, and
    \[ g \big ( 1 - e^{-3 t}, 4) \big) = \sup_{\tau \in (0,1)} \frac{e^{- 6 t} \tau}{\tau + 1 - \exp(-3t)} \Big ( 1 - \frac{9 \tau (1 - \ln \tau)}{9 \ln 9 - 8}\Big ) \pl. \]
    When $t = 10^{-3}$, $g \geq 0.81$ with $\tau = 0.0302$. When $t = 0.01$, $g \geq 0.54$ with $\tau = 0.0980$. When $t = 0.1$, $g \geq 0.14$ with $\tau = 0.2590$.  When $t = 1$, $g \geq 3*10^{-4}$ with $\tau = 0.4187$.
\end{exam}

\subsection{Classical and Semiclassical Converse Bounds}
In this Subsection, we upper bound mutual information decay for a wide range of continuously parameterized channels on commuting density matrices and classical probability vectors, particularly including all channels formed by composing quantum Markov semigroups at infinitesimal times as long as each has an invariant state. In Section \ref{sec:sudden}, as illustrated by Figure \ref{fig:visual}, we use the fact that two quantum states can be arbitrarily close in trace distance and similar distance measures, yet neither is a convex combination involving the other. In particular, the density $\rho_{\theta,0}$ as in Example \ref{exam:basic} can be arbitrarily close to $\ketbra{0}$ for sufficiently small $\theta$, yet two distinct pure states will never be convex combinations involving each other. In contrast, sets of classical probability distributions or commuting quantum densities have only orthogonal pure states (elements of the same basis). Hence closeness eventual implies comparability in the semidefinite Loewner order. In this Subsection, we exploit this observation to yield decay converse bounds assuming simultaneous diagonalizability of all densities involved.

The culminating result of this Subsection is Theorem \ref{thm:classical}, a technical version of Theorem \ref{thm:converse} part (2) with explicit constants. We briefly summarize a general strategy of proof:
\begin{itemize}
    \item Show that if $\rho$ and $\sigma$ are sufficiently close and commute, then they're order-comparable.
    \item Show that if $\rho$ is order-comparable to $\sigma$, then sudden decay does not happen.
    \item Show that if $\rho$ and $\sigma$ are sufficiently far, then additive bounds rule out sudden decay.
\end{itemize}
First, the essential distinction of the commuting from the general case:
\begin{lemma} \label{lem:commorder}
    For commuting densities $\rho$ and $\sigma$, $\rho \geq \sigma - \| (\rho - \sigma)_{\text{supp}(\sigma)}\|_\infty \hat{1} |_{\text{supp}(\sigma)}$.
\end{lemma}
\begin{proof}
Without loss of generality, we may apply permutations to re-order eigenvalues in decreasing order for $\sigma$. Projected to the support of $\sigma$,
\begin{equation}
    | \sigma - \rho |_{\text{supp}(\sigma)} = \Bigg | \begin{pmatrix}
        \sigma_{1,1} & & 0 \\ & ... & \\ 0 & & \sigma_{m,m} \end{pmatrix} 
        -  \begin{pmatrix}
        \rho_{1,1} & & 0 \\ & ... & \\ 0 & & \rho_{m,m} \end{pmatrix} \Bigg |
    \leq \| (\rho - \sigma)_{\text{supp}(\sigma)} \|_{\infty} \id |_{\text{supp}(\sigma)} \pl.
\end{equation}
Outside the support of $\sigma$, $(\sigma - \rho)_{\text{supp}(\perp \sigma)} \leq 0$. Hence
\begin{equation*}
    \rho = \sigma - (\sigma - \rho) \geq \sigma - \| (\rho - \sigma)_{\text{supp}(\sigma)}\|_\infty \hat{1} |_{\text{supp}(\sigma)} \pl.
\end{equation*}
\end{proof}
\begin{lemma} \label{lem:origcompare}
Let $\epsilon, \zeta \in (0,1)$. Let $\rho, \sigma$, and $\omega$ be densities such that $\rho \geq (1-\zeta) \sigma$. Then
    \[ D(\rho \| \sigma) \leq \frac{1}{(1-\epsilon)^2}
        \Big ( 1 + \epsilon \Big ( \frac{g_{P_{\sigma}(\omega) | \sigma}}{1 - \zeta} -1 \Big ) \Big )
      D((1-\epsilon) \rho + \epsilon \omega \| (1-\epsilon) \sigma + \epsilon \omega ) \pl, \]
where $P_\sigma$ denotes the support projection of $\sigma$, and $g_{P_{\sigma}(\omega) | \sigma} := \min_{g} \{ g | g \sigma \geq P_{\sigma}(\omega)\}$.
\end{lemma}
\begin{proof}
Let $\tilde{\rho}$ be such that $\rho = (1 - \zeta) \sigma + \zeta \tilde{\rho}$. Again recall Equation \eqref{eq:integralform}. For any norm,
\begin{equation} \label{eq:trivratio1}
\| (1-\epsilon) \rho + \epsilon \omega - (1-\epsilon) \sigma - \epsilon \omega \|
    =  (1-\epsilon) \| \rho - \sigma \| \pl.
\end{equation}
Let $\eta := P_{\sigma}(\omega)$ as shorthand. Then
\[ (1-\epsilon) (\rho, \sigma)_t + \epsilon \eta \leq \Big ( 1 + \epsilon \Big ( \frac{g_{P_{\sigma}(\omega) | \sigma}}{1 - \zeta} -1 \Big ) \Big ) (\rho, \sigma)_t \pl. \]
Our modification of \cite[Lemma 2.1]{gao_complete_2022} then implies the corresponding inequality on weighted norms:
\[ \|\rho - \omega\|_{((1-\epsilon) (\rho, \sigma)_t + \epsilon \eta)^{-1}}^2 \geq 
     \Big ( 1 + \epsilon \Big ( \frac{g_{P_{\sigma}(\omega) | \sigma}}{1 - \zeta} -1 \Big ) \Big ) \|\rho - \omega\|_{(\rho, \sigma)_t^{-1}}^2 \]
To eliminate the ``$t$" in the denominator, we overestimate $1/(1 - \zeta t) \leq 1/(1-\zeta)$, since $t \leq 1$. Combining with Equation \eqref{eq:trivratio1}, then integrating Equation \eqref{eq:integralform} completes the Lemma.
\end{proof}

\begin{lemma} \label{lem:casecombine}
Let $\sigma$ and $\rho$ be commuting densities. Recall the function $f_{\tilde{m}}(\epsilon)$ as in Corollary \ref{cor:almostconcave}, and let $\tilde{m}$ be the smallest non-zero eigenvalue of $\sigma \oplus \E(\sigma)|_{\text{supp}(\sigma)}$. For any $a \in (0,1)$ satisfying $a \tilde{m}^2 > 2 f_{\tilde{m}}(\epsilon) / (1-\epsilon)$,
\begin{equation*}
\begin{split}
& D((1-\epsilon) \rho + \epsilon \omega \| (1-\epsilon) \sigma + \epsilon \omega) \\ & \pl \geq D(\rho \| \sigma)
    \begin{cases}
     (1- \epsilon - 2 f_{\tilde{m}}(\epsilon) / a \tilde{m}^2) & : D(\rho \| \sigma) \geq a \tilde{m}^2/2 \\
        (1 - a) (1-\epsilon)^2 / \big ( (1 - \epsilon) (1 - a) + \epsilon \min_{g} \{ g | g \sigma \geq P_{\sigma}(\omega)\} \big ) & : D(\rho \| \sigma) \leq a \tilde{m}^2/2 
    \end{cases}
\end{split}
\end{equation*}
for all densities $\rho$ commuting with $\sigma$ and $\omega$.
\end{lemma}
\begin{proof}
We first split the Lemma into two cases: (1) $D(\rho \| \sigma)$ is not too small, and (2) $D(\rho \| \sigma)$ is not too large.

\noindent \textbf{Case (1)}: since $\rho$ and $\sigma$ were assumed to commute, their relative entropy is equivalent to the classical Kullback-Leibler divergence. Since the relative entropy is lower-bounded, a sufficiently small additive bound is effectively a multiplicative bound. Recall Corollary \ref{cor:almostconcave} as based on results from \cite{bluhm_continuity_2022}. Recalling the notation therein,
\begin{equation*}
\begin{split}
D((1-\epsilon) \rho + \epsilon \omega \| (1-\epsilon) \sigma + \epsilon \omega)
    \geq & (1 - \epsilon) D(\rho \| \sigma)
    - f_{\tilde{m}}(\epsilon) \pl.
\end{split}
\end{equation*}

\noindent \textbf{Case (2)}: by Lemma \ref{lem:commorder}, $ \rho \geq \sigma - \| \rho - \sigma\|_\infty \hat{1} |_{\text{supp}(\sigma)} \geq (1 - \|\rho - \sigma\|_\infty / \tilde{m}) \sigma$. By Lemma \ref{lem:origcompare}, 
\[ D((1-\epsilon) \rho + \epsilon \omega \| (1-\epsilon) \sigma + \epsilon \omega )
    \geq \frac{(1- \| \rho - \sigma\|_\infty / \tilde{m})(1-\epsilon)^2 D(\rho \| \sigma)}{(1 - \epsilon) (1 - \|\rho - \sigma\|_\infty / \tilde{m}) + \epsilon \min \{ g | g \sigma \geq P_{\sigma}(\omega)\}}  \pl.
      \]
Via Pinsker's inequality and Jensen's inequality, $\| \rho - \sigma \|_\infty^2 \leq \sqrt{2 D(\rho \| \sigma)}$, so
\[ D((1-\epsilon) \rho + \epsilon \omega \| (1-\epsilon) \sigma + \epsilon \omega)
    \geq \frac{(1 - \sqrt{2 D(\rho \| \sigma)} / \tilde{m}) (1-\epsilon)^2 D(\rho \| \sigma)}{(1 - \epsilon) (1 - \sqrt{2 D(\rho \| \sigma)} / \tilde{m}) + \epsilon \min \{ g | g \sigma \geq P_{\sigma}(\omega)\}} \pl. \]

\noindent \textbf{Combining cases}: assume that $D(\rho \| \sigma) \leq a \tilde{m}^2 / 2$ for some $a \in (0,1)$. Then by Case (2), 
\[ D((1-\epsilon) \rho + \epsilon \omega \| (1-\epsilon) \sigma + \epsilon \omega)
    \geq \frac{(1 - a) (1-\epsilon)^2 }{(1 - \epsilon) (1 - a) + \epsilon \min \{ g | g \sigma \geq P_{\sigma}(\omega)\}} D(\rho \| \sigma) \pl. \]
This bound is multiplicative and non-trivial. In contrast, if we assume that $b (1-\epsilon) D(\rho \| \sigma) \geq f_{cty}(\epsilon)$ for some $b \in (0,1)$, then Case (1) implies that
\[ D((1-\epsilon) \rho + \epsilon \omega \| (1-\epsilon) \sigma + \epsilon \omega) \geq (1-\epsilon)(1-b) D(\rho \| \sigma) \pl. \]
For both assumptions to potentially be satisified for the same $\rho$ requires that $a b \tilde{m}^2 \geq 2 f_{\tilde{m}}(\epsilon) / (1-\epsilon)$. In this regime, can combining cases yields a multiplicative bound for all $\rho$ and sufficiently small $\epsilon$. For efficiency, we may assume the aforementioned inequality saturates, obtaining that $b = 2 f_{\tilde{m}}(\epsilon) / (1-\epsilon) a \tilde{m}^2$.
\end{proof}
\begin{rem} \normalfont
Fixing $\tilde{m}$ and the base so that $\log = \ln$ one may estimate
\[ f_{\tilde{m}}(\epsilon) \leq \epsilon \Big ( \ln \Big ( \frac{1}{\tilde{m}} \Big ) + \frac{1}{\tilde{m}} - \ln \epsilon + 1 \Big ) \pl. \]
First,
\[ \epsilon \log \Big ( \epsilon + \frac{1-\epsilon}{\tilde{m}} \Big ) \leq \epsilon \log \Big ( \frac{1}{\tilde{m}} \Big ) \pl. \]
Since $\epsilon$ is assumed small, we probably do not lose too much in this estimate. Second,
\[ (1-\epsilon) \ln \Big (1 - \epsilon + \frac{\epsilon}{\tilde{m}} \Big ) \leq \frac{\epsilon}{\tilde{m}} \pl. \]
Again, with $\epsilon$ sufficiently small, this estimate becomes tight. Third, splitting $h(\epsilon)$ into individual terms, we keep the term $\epsilon \log (1/\epsilon)$ as is. Fourth,
\[ - (1-\epsilon) \ln ( 1-\epsilon ) \leq \frac{(1-\epsilon) \epsilon}{1-\epsilon} = \epsilon \pl. \]
\end{rem}

\begin{theorem} \label{thm:classical}
    Let $\L$ a Lindbladian generating semigroup $(\Phi^t)$ with fixed point subspace projection $\E$, assuming $c \E \geq_{cp} \Id$. Let $\tilde{m}$ denote the smallest non-zero eigenvalue of $\sigma \oplus \E(\sigma)|_{\text{supp}(\sigma)}$ for a density $\sigma$, and $\epsilon_{\L, t} := \exp(- t c \|\L\|_{\Diamond} / 2)$ for $t \in \RR^+$. Recall
    \[ f_{\tilde{m}}(\epsilon) := h(\epsilon) 
    + \epsilon \log(\epsilon + (1 - \epsilon) \tilde{m}^{-1})
    + (1 - \epsilon) \log((1 - \epsilon) + \epsilon \tilde{m}^{-1}) \]
    as in Corollary \ref{cor:almostconcave}. Then for any $a \in (0,1)$ satisfying $a \tilde{m}^2 > 2 f_{\tilde{m}}(\epsilon_{\L, t}) / (1-\epsilon_{\L, t})$,
    \begin{equation*}
        \begin{split}
        D(\Phi^t(\rho) \| \Phi^t(\sigma)) \geq D(\rho \| \sigma)
            \begin{cases}
             (1- \epsilon_{\L, t} - 2 f_{\tilde{m}}(\epsilon_{\L, t} / a \tilde{m}^2 ) & : D(\rho \| \sigma) \geq a \tilde{m}^2/2 \\
                (1 - a) (1-\epsilon_{\L, t})^2 \big ( (1 - \epsilon_{\L, t} (1 - a) + \epsilon_{\L, t} \tilde{g} \big ) & : D(\rho \| \sigma) \leq a \tilde{m}^2/2 
            \end{cases}
        \end{split}
        \end{equation*}
    for all $\rho$ as long as $\E(\rho) = \E(\sigma)$ and $\rho, \sigma$ and $\E(\rho)$ all commute, where $\tilde{g} := \min_{g} \{ g | g \sigma \geq P_{\sigma}(\omega)\}$.
\end{theorem}
\begin{proof}
First, we apply Lemma \ref{lem:scombine} to determine that
\begin{equation} \label{eq:phirep1}
D(\Phi^t (\rho) \| \Phi^t (\sigma)) \geq
    D(e^{- t c \|\L\|_{\Diamond} \S_{\text{rep}} / 2}(\rho) \| e^{- t c \|\L\|_{\Diamond} \S_{\text{rep}} / 2}(\sigma)) \pl,
\end{equation}
where $S_{\text{rep}}(\rho) = \rho - \E(\rho)$, and $S_{\text{rep}}(\sigma) = \sigma - \E(\sigma)$. The same ``$S_{\text{rep}}$" holds in both arguments via the assumption that $\E(\rho) = \E(\sigma)$. Expanding,
\[ e^{- t c \|\L\|_{\Diamond} S_{\text{rep}} / 2}(\rho) = \big (1 - e^{- t c \|\L\|_{\Diamond} / 2} \big ) \rho + e^{- t c \|\L\|_{\Diamond} / 2} \E(\rho) \pl, \]
and similarly for $\sigma$. Apply Lemma \ref{lem:casecombine}.
\end{proof}
\begin{cor} \label{cor:classicalmut}
Let $\L, (\Phi^t), \E$ and $(\epsilon_{\L, t})$ be as in Theorem \ref{thm:classical}. Assume that $(\Phi^t)$ acts on subsystem $B$ of the bipartite system $AB$ such that $(\Id^A \otimes \E)(\rho) = \rho^A \otimes \omega$ for some density $\omega$. Let $\tilde{m}$ be the smallest non-zero eigenvalue of $(\rho^A \otimes \rho^B) \oplus (\rho^A \otimes \omega|_{\text{supp}(\rho^A \otimes \rho^B)})$, and $f_{\tilde{m}}$ and be as in Theorem \ref{thm:classical}. Then
\begin{equation*}
\begin{split}
& I[A:B]((\Id^A \otimes \Phi^t)(\rho)) \\ & \pl \geq I[A:B](\rho) 
    \begin{cases}
     (1- \epsilon_{\L, t} - 2 f_{\tilde{m}}(\epsilon_{\L, t} / a \tilde{m}^2 ) & : I[A:B](\rho) \geq a \tilde{m}^2/2 \\
        (1 - a) (1-\epsilon_{\L, t})^2 \big ( (1 - \epsilon_{\L, t} (1 - a) + \epsilon_{\L, t} g \big ) & : I[A:B](\rho) \leq a \tilde{m}^2/2 
    \end{cases} \pl,
\end{split}
\end{equation*}
where $g := \min \{\gamma | \gamma \rho^B \geq P_{ \rho^B}(\omega^B)\}$.
\end{cor}
\begin{proof}
Recall that $I[A:B](\rho) = D(\rho^{AB} \| \rho^A \otimes \rho^B)$, then apply Theorem \ref{thm:classical}.
\end{proof}
For the mutual information as in Corollary \ref{cor:classicalmut}, uniformity in ``$\rho$" is replaced by the bound only depending on $\rho$'s marginals, rather than on the original, correlated density.

\subsection{Starting from Already Decayed States}
Once relative entropy has started to decay via convex, partial replacement by another state, further replacement may be more gradual:
\begin{prop} \label{prop:4decay}
For densities $\rho, \theta, \sigma, \omega$ such that $\theta \leq c \omega$, and for $\epsilon, \zeta \in (0,1)$ such that $\epsilon > \zeta$,
\[ D((1 - \epsilon) \rho + \epsilon \theta \| (1-\epsilon) \sigma + \epsilon \theta) \geq \frac{\zeta}{c \epsilon} \Big ( \frac{1-\epsilon}{1-\zeta} \Big )^2 D((1 - \zeta) \rho + \zeta \omega \| (1 - \zeta) \sigma + \zeta \omega) \pl. \]
\end{prop}
\begin{proof}
We observe immediately that for $t \in [0,1]$,
\[ (1 - \epsilon) (\rho, \sigma)_t + \epsilon \theta \leq \frac{c \epsilon}{\zeta} \big ( (1 - \zeta) (\rho, \sigma)_t + \zeta \omega \big ) \]
Hence Lemma \ref{lem:normcomp} implies that
\[ \| (1-\epsilon)(\rho - \sigma) \|_{((1 - \epsilon) (\rho, \sigma)_t + \epsilon \theta)^{-1}}^2 
 \geq  \frac{\zeta}{c \epsilon} \Big ( \frac{1-\epsilon}{1-\zeta} \Big )^2 \| (1-\zeta) (\rho - \sigma)  \|_{\big ( (1 - \zeta) (\rho, \sigma)_t + \zeta \omega \big )^{-1}}^2 \pl. \]
The 2nd derivative form of relative entropy as in Equation \eqref{eq:integralform} then completes the Lemma.
\end{proof}
A lingering question is whether a state that has undergone any initial decay admits a universal decay converse in terms of the amount of that initial decay. In Example \ref{exam:basic}, starting from a pure state appeared essential to achieve an arbitrarily large ratio of initial to final mutual information, which is generalized by Proposition \ref{prop:4decay}. The same should apply to the construction in Theorem \ref{thm:group}, as there it was shown that any decay yields a convex combination with the fixed point projection. However, it is less clear if Lindbladians generally yield convex combinations with a fixed point subalgebra after finite time. Such a result would have broader implications, so it is left to future study. If the answer is no, then one might conjecture the existence of semigroups and other parameterized channels that avoid the regime of Proposition \ref{prop:4decay}.

{\bibliography{Converse}
\bibliographystyle{unsrt} }

\end{document}